\documentclass[11pt, reqno]{amsart}
\usepackage{physics}
\usepackage{hyperref}
\usepackage{amsmath}
\usepackage{amsthm}
\usepackage{amsfonts}
\usepackage{amssymb}
\usepackage{dsfont}
\usepackage{enumerate}
\usepackage[english=british]{csquotes}
\usepackage{geometry}
\newtheorem{proposition}{Proposition}
\newtheorem{theorem}{Theorem}
\newtheorem{lemma}{Lemma}
\newtheorem{corollary}{Corollary}

\newtheorem{definition}{Definition}

\renewcommand{\var}{\text{Var}}
\newcommand{\cov}{\text{Cov}}
\renewcommand{\vec}{\text{vec}}
\newcommand{\svec}{\text{svec}}

\begin{document}
\title{Gaussian Dynamical Quantum State Tomography}
\author[Rall]{Hjalmar Rall$^{1,2}$}
\email{hjalmar.rall@tum.de}
\address{$^1$ Department of Mathematics, Technical University of Munich}
\address{$^2$ Munich Center for Quantum
Science and Technology (MCQST),  M\"unchen, Germany}
\begin{abstract}
Standard quantum state tomography assumes sufficient control of a system to measure an informationally complete set of observables. Dynamical quantum state tomography (DQST) presents an alternative: given a system with known dynamics and a single fixed observable, it almost always suffices to control only the time at which each i.i.d. copy of the system is measured. This work presents an analogous scheme for tomography of multi-mode Bosonic Gaussian states undergoing Gaussian evolution, using a fixed single-mode homodyne measurement and only assuming control of the time of measurement. I prove that the scheme enables tomography for all discrete homogenous Gaussian evolutions and Gaussian quantum dynamical semigroups except for a null set which includes unitary evolution. When the state is known to be pure, a smaller number of measurement times is shown to be sufficient.
\end{abstract}

\maketitle
\section{Introduction}
A fundamental quantum information task is that of determining the state of a quantum system when little to none of it's properties are known. This is usually achieved through quantum state tomography, requiring an informationally complete set of measurements performed on identical independent (i.i.d.) copies of the state. The Heisenberg picture reveals that such a set of observables is equivalent to copies of a single fixed observable that have undergone different (dual) evolutions. This raises the natural question of whether it is sufficient to consider a single repeated evolution or a one-parameter family of evolutions, and in particular when it is possible to perform tomography of a state undergoing a known homogeneous evolution by performing a fixed measurement on i.i.d copies whilst varying only the time of measurement. This scheme is known as Dynamical Quantum State Tomography (DQST) and has been studied in \cite{peruzzo2025reconstructing, merkel2010random}. One important motivation for DQST is the practical challenges of performing standard quantum state tomography of multipartite systems where the exponential number of measurement settings and the need for multipartite measurements presents a challenge.\\

The problem was studied in \cite{kech2016dynamical}, where it was shown that DQST is generically possible for a \(d\)-dimensional state measured with a \(d\)-outcome POVM after the state undergoes discrete-time unitary evolution or with a binary POVM if general discrete-time evolution is allowed. In \cite{merkel2010random} the authors present a quantum state reconstruction technique for a \(d\)-dimensional Hilbert space wherein an observable in \(\mathfrak{su}(d)\) is evolved under repeated application of a fixed random unitary. Despite not being informationally complete, their scheme approximately determines the initial state with high fidelity. In \cite{xiao2024quantum} DQST is formulated in the language of classical control theory for closed systems undergoing known discrete-time unitary dynamics.\\

A key feature of DQST is that global dynamics in combination with local measurements can in some cases enable tomography of the global state, and as such it is related to the quantum marginal problem. This is studied in \cite{peruzzo2024reconstructing} in the context of finite-dimensional systems under discretized Markovian dynamics. They show that randomly choosing parameters of the evolution enables unique determination of the state with probability 1, and formulate an optimization problem to estimate the state in the case of non-ideal measurements.\\

In \cite{rall2025quantum}, dynamical quantum tomography with binary measurement outcomes is studied for both states and observables on \(d\)-dimensional systems under general evolution. Such tomography is shown to be generically possible, though unitary evolution is insufficient depending on the dimension and whether a state or observable is to be determined. General results for infinite time series are derived and used to adapt the main result to continuous-time Markovian evolution with binary measurements at arbitrary times.  \\

In this work I study DQST for Bosonic systems under the assumption that both the state and evolution are Gaussian. In this case, the states have a finite-dimensional representation and may be uniquely identified with a finite set of informationally complete measurements. The primary motivation is to extend some of the results of \cite{rall2025quantum} to continuous variable systems. Nevertheless, it is conceivable that there are practical scenarios favouring the scheme presented here over standard homodyne \cite{vrehavcek2009effective, d1999universal, d2007homodyne, esposito2014pulsed} or heterodyne \cite {vrehavcek2015surmounting, bittel2025energy, chapman2022bayesian} state tomography. In photonic systems for instance, Gaussian DQST (GDQST) could be realised in a manner similar to a photonic quantum walk \cite{graefe2016integrated} using just a fixed photonic circuit, an optical delay line, and a switch in order to implement a Gaussian channel multiple times. This is potentially simpler than homodyne or heterodyne tomography which must make a tradeoff between the number of circuit configurations and the number of detectors and cannot determine a multi-mode state using only one of each. Whilst the number of measurement settings required to determine a Gaussian state scales only as the square of the number of modes of the system, this reduction in complexity could conceivably prove advantageous. \\

The article is laid out as follows: 
In section 2, I introduce notation and the machinery for working with Bosonic Gaussian states represented in terms of their first and second moments, and discuss homodyne measurement in the context of quantum state tomography. In section 3, I present a version of the result of \cite{rall2025quantum} which states that the power sequence of a finite-dimensional matrix is almost always fully determined by a finite subsequence, and show that this applies to DQST of Gaussian states undergoing general Gaussian discrete-time evolution. In section 4, this is extended to DQST of Gaussian states under continuous-time Markovian dynamics. The main results are given in section 5 where the requirements are given for DQST of any Gaussian state to be possible under almost any Gaussian evolution, as well as a No-Go result for unitary evolution. Lastly, in section 5.1 it is shown that fewer measurements are required given prior information that the state is pure.

\section{Notation and preliminaries}
Henceforth, I shall denote by \(\mathbb{M}_n(\mathds{R})\) the set of \(n \times n\) real matrices, by \(\mathbb{S}_n(\mathds{R})\) the set of \(n \times n\) real symmetric matrices, and by \(\lor^2 \mathds{R}^n\) the symmetric subspace of \(\mathds{R}^n \otimes \mathds{R}^n\). Furthermore, \(\mathcal{I}_n\) shall denote the index set consisting of unordered pairs \(\{i, j\}\) with \(i,j \in \mathds{Z}^n\).

\begin{definition}
  Let \(\ket{\psi_{\{i,j\}}} := \frac{1}{\sqrt{2(1+\delta_{ij})}} \bigg(\ket{i}\ket{j} + \ket{j}\ket{i}\bigg)\) for \(\{i,j\} \in \mathcal{I}_n\). These define an orthonormal basis of \(\lor^2 \mathds{R}^n\). \end{definition}

\begin{definition}
  For any symmetric matrix \(S \in \mathbb{S}_{n}(\mathds{R})\), define the symmetric vectorization \cite{schaecke2004kronecker} \(\) as
  \begin{equation*}
    \begin{split}
 \mathds{R}^{\frac{n(n+1)}{2}} \ni     \svec(S) &:= (s_{11}, \sqrt{2}s_{21}, \hdots, \sqrt{2} s_{n1}, s_{22}, \sqrt{2}s_{32}, \hdots, \sqrt{2}s_{n2}, \hdots, s_{nn})\\
      &= \sum_{\{i,j\} \in \mathcal{I}_n} (\sqrt{2}(1-\delta_{ij}) + \delta_{ij}) \bra{j}S\ket{i}\ket{\psi_{\{i,j\}}}
    \end{split}
  \end{equation*}
  where, following convention, the additional factors of \(\sqrt{2}\) are added to ensure that \(\Tr(ST) = \svec(S)^T\svec(T) \) for \(\forall S,T\in \mathbb{S}_n(\mathds{R})\).
\end{definition}

\begin{definition}
  Let \(S \in \mathbb{S}_{2m}(\mathds{R})\). The \textit{symmetric Kronecker product} \cite{schaecke2004kronecker}, is defined for any \(A, B \in \mathbb{M}_{2m }(\mathds{R})\) and \(S \in \mathbb{S}_{2m}(\mathds{R})\) by
  \begin{equation*}
(A \otimes_s B)\svec(S) = \frac{1}{2}\svec(ASB^T + BSA^T).
  \end{equation*}
  Equivalently, let \(Q\) be the unique \(m(2m+1) \times (2m)^2 \) matrix with \(Q \vec(S) = \svec(S) \) and \( Q^T \svec(S)  = \vec(S) \). Then the symmetric Kronecker product is defined by
  \begin{equation*}
(A \otimes_s B) = \frac{1}{2}Q(A \otimes B + B \otimes A)Q^T.
  \end{equation*}
\end{definition}
It will be useful to note that the symmetric kronecker product has the following property which follows directly from the definition:
\begin{equation}\label{symmkronpowers}
(A^i \otimes_s A^i) = (A \otimes_s A)^i.
\end{equation}
 
\subsection{Gaussian dynamics}
Consider an \(m\)-mode bosonic system with phase space having coordinates \((q_1, \hdots, q_m, p_1, \hdots, p_m) \in \mathds{R}^{2m},\) and symplectic form 
\(
\Omega := \begin{pmatrix}
  0 & I_m\\
  -I_m & 0\\
\end{pmatrix}.
\)
Let \(Q_i, P_i\) be the canonical quadrature operators on the \(i\)-th mode and collect them in a vector \(\mathbf{R} := (R_1, \hdots, R_{2m}) = (Q_1,\hdots, Q_m, P_1, \hdots, P_m)\), then a state \(\rho\) has first and second moments given by the displacement vector \(\mathbf{d} \in \mathds{R}^{2m}\) with \(d_k = \Tr \rho R_k \) and covariance matrix \(\Gamma \in \mathbb{M}_{2m}(\mathds{R})\) with \(\Gamma_{kl} = \Tr [\rho(\{R_k - d_k \mathds{1}, R_l - d_l \mathds{1}\}_+)]\) respectively. A state is said to be Gaussian if it has a Gaussian Wigner distribution
\begin{equation}
\mathcal{W}(\xi) = c e^{-(\xi - \mathbf{d})\Gamma^{-1}(\xi - \mathbf{d})}
\end{equation}
and thus every Gaussian state is uniquely determined by the pair \((\Gamma, \mathbf{d})\).\\

A Gaussian channel is a quantum channel mapping Gaussian states to Gaussian states. In terms of the above parameterization of Gaussian states, there is a one-to-one correspondence between Gaussian channels and mappings
\begin{equation}
 (\Gamma, \mathbf{d}) \mapsto (X \Gamma X^T + Y, X^T \mathbf{d})
\end{equation}
with \(X\in \mathbb{M}_{2m }(\mathds{R}), Y \in \mathbb{S}_{2m}(\mathds{R})\) satisfying the symplectic condition
\begin{equation}\label{symplecticcond}
Y + i\Omega - iX \Omega X^T \geq 0.
\end{equation}
Denote by \(\mathcal{T}_G\) the set of all pairs \((X, Y)\) parameterizing a Gaussian channel.\\
Throughout this work, \(\Gamma\) will be used interchangeably with its symmetric vectorization \(\ket{\gamma} := \svec(\Gamma) \) for which \(\svec(X \Gamma X^T) = X \otimes_s X \ket{\gamma } \).

\begin{proposition}\label{measure}
        Define an isomorphism \(\nu : \mathbb{M}_{2m}(\mathds{R}) \times \mathbb{S}_{2m}(\mathds{R}) \rightarrow \mathds{R}^{m(6m+1)}\) such that \(\nu (X, Y) = \vec(X) \otimes \svec(Y)\). The pushforward \((\nu^{-1})_{*}\lambda_{m(6m+1)}\) of the Lebesgue measure on \(\mathds{R}^{m(6m+1)}\) defines a measure on \(\mathbb{M}_{2m}(\mathds{R}) \times \mathbb{S}_{2m}(\mathds{R})\) and hence (by restriction) a measure on \(\mathcal{T}_G\). The restricted measure is non-trivial since the Euclidean dimension \(\text{dim}(\nu(\mathcal{T}_G)) = m(6m+1)\).
\end{proposition}

\begin{definition}[Null set with respect to the Gaussian channels]
A subset of \(\mathcal{T}_G\) is referred to as a null set if it has zero measure with respect to the measure defined in Proposition \ref{measure}.
\end{definition}

\begin{lemma}
  The subset of \(\mathcal{T}_G\) for which the matrices \(X \in \mathbb{M}_{2m }(\mathds{R})\) are either defective or have a non-trivial kernel is a null set in \(\mathcal{T}_G\).
\end{lemma}
\begin{proof}
For a general Gaussian channel \(Y\) can always be chosen to satisfy the symplectic condition \eqref{symplecticcond}, so \(X\) can be an arbitrary real matrix, for which these are standard results. The second may be observed by noting that \(det(X) = 0\) is a polynomial in the entries and therefore defines a real algebraic variety. It must be satisfied either on the whole of \(\mathbb{M}_{2m}(\mathds{R})\), which is clearly not the case, or on a null set. Hence the set of real matrices with non-trivial kernel form a null set. Similarly the set of \(X\) with degenerate eigenvalues is determined by the zero set of the discriminant of the characteristic polynomial of \(X\) which is again a real algebraic variety and must be a null set since the set of non-degenerate matrices is not empty \cite{rall2025quantum}.
\end{proof}

\subsection{Homodyne measurement}
For homodyne tomography (e.g. \cite{sych2012informational}) one wishes to measure an informationally complete POVM made up of projectors onto the quadratures \(X(\theta) = cos(\theta) Q + sin(\theta) P\), corresponding to the eponymous homodyne measurements.
Gaussian states are described fully by their first and second moments and can therefore be determined by a finite number of homodyne measurements. For the case of a multimode Gaussian state it is desirable to measure only a single mode due to the practical challenges of multimode measurements; this may be achieved by acting on the multimode state with some Gaussian unitary prior to measurement. \v Reha\v cek et al. \cite{vrehavcek2009effective} describe a practical experimentally realisable scheme whereby all the first and second moments may be obtained using \(m\) beamsplitter configurations and a phase shifter on the measured mode scanning over \(\theta\). This paper addresses the question of whether these multiple configurations may be replaced by a single known Gaussian channel applied a varying number of times.\\

Suppose W.L.O.G. that one physically measures the single-mode operator \(\hat{R}_1\). Assuming a perfect detector this yields \(\langle \hat{R_1} \rangle \) and \(\text{Var}(\hat{R_1})\), corresponding to \(d_1\) and \(\Gamma_{11}\) respectively. In the interest of generality, one can furthermore consider the scenario where one has access not just to a fixed single-mode measurement, but also to a fixed passive Gaussian unitary. Applying a passive Gaussian unitary \(U\) prior to measurement, the effective measurement \(U^{\dag}\hat{R}_1U\) corresponds to \(\bra{1}O \ket{d}\) and \(\bra{11} O \otimes O \ket{\gamma}\) with \(O\) orthogonal symplectic. Hence \(\mathbf{b} \equiv \ket{b} = O^T \ket{1}\) can be any unit vector. In this way it is possible to measure linear combinations \(\sum_j b_{j} R_j = \mathbf{b}^T \mathbf{R}\) of quadrature variables yielding
\begin{equation}
  \langle \mathbf{b}^T \mathbf{R} \rangle = \mathbf{b}^T \mathbf{d} = \braket{b}{d}\\
\end{equation}
and
\begin{equation}\label{eq:varaTR}
  \var(\mathbf{b}^T \mathbf{R}) = \vec(\mathbf{b}\mathbf{b}^T)^T \vec(\Gamma) = \svec(\mathbf{b}\mathbf{b}^T)^T \svec(\Gamma).
\end{equation}
Eq. \eqref{eq:varaTR} follows directly from the identity  \(\var(\sum_{j=1}^{2m} b_{j} R_j) = \sum_{j, k = 1}^{2m}  b_{j} b_{k} \cov(R_j, R_k)\).\\
For notational convenience, define \(\ket{a} := \svec(\mathbf{b}\mathbf{b}^T)\) so that Eq. \eqref{eq:varaTR} becomes
\begin{equation}
\var(\mathbf{b}^TR) = \braket{a}{\gamma}
\end{equation}
\begin{proposition}
  Any set of \(m(2m+1)\) homodyne measurement settings apart from a null set is sufficient for tomographic reconstruction of an \(m\)-mode Gaussian state.
\end{proposition}
\begin{proof}
  In general, a homodyne measurement yields statistics for a linear combination \(\sum_j b_{j} R_j = \mathbf{b}^T \mathbf{R}\) of quadrature variables \(R_j\). A set of measurements defined by \(\{\mathbf{b}_1, ..., \mathbf{b}_{m(2m+1)}\}\) yield a system of equations
  \(\var(\mathbf{b}_i^T \mathbf{R}) =  \svec(\mathbf{b}_i\mathbf{b}_i^T)^T \svec(\Gamma)\), which has a unique solution if and only if there is a nonzero vector \(\mathbf{v}\) such that \(\sum_i v_i b_{ij}b_{ik} = 0, \: \forall \{j,k\} \in \mathcal{I}_{2m}\). Simultaneously, these measurements yield a system of equations \(\langle \mathbf{b}_i^T \mathbf{R} \rangle = \mathbf{b}_i^T \mathbf{d}\) which has a unique solution if and only if there exists a non-zero vector \(\mathbf{u}\) such that \(\sum_i u_i b_{ij} = 0, \: \forall j = \{1, \hdots, 2m\}\). These define algebraic varieties, which are either zero on a null set or zero everywhere. To exclude the latter, choose \(\mathbf{b}_i\) to be the elements of \(\{e_j + e_k: \{j,k\} \in \mathcal{I}_{2m}\} \) which spans \(\mathds{R}^{2m}\). The \(\mathbf{b}_i\mathbf{b}_i^T\) make up the set \(\{E_{jj} + E_{kk} + E_{jk} + E_{kj}: \{j,k\} \in \mathcal{I}_{2m}\}\) so that \(\svec(\mathbf{b}_i\mathbf{b}_i^T)\) span \(\mathds{R}^{m(2m+1)}\). Thus there exists at least one measurement setting where both \(u\) and \(v\) are non-zero.
\end{proof}

It is reasonable, therefore, to assume that \(m(2m+1)\) fixed measurement settings may be replaced by a single measurement setting preceeded by a homogeneous Gaussian evolution, the duration of which is varied. 

\section{Time Series }
Suppose one is given independent copies of an unknown \(m\)-mode Gaussian state \((\Gamma,d)\), and may apply a known discrete-time gaussian evolution (Gaussian channel) \((X,Y)\) arbitrarily many times before performing a fixed homodyne measurement \(\mathbf{b}^T\mathbf{R}\). Repeated application of the channel \((X,Y)\) results in a sequence \((\Gamma_i, d_i)_{i=0}^n\) with \(\Gamma_i = (X^T)^i \Gamma X^i + \sum_{j=0}^{i-1} (X^T)^j Y X^j\) and \(\mathbf{d}_i = (X^T)^i \mathbf{d}\).
Now consider the sequence \(\Gamma_i'  = (X^T)^i \Gamma X^i\). If \(X\) and \(Y\) are both known, then  it is clear that \((\Gamma_i, \mathbf{d}_i)_{i=0}^n\) can always be determined from \((\Gamma_i' , \mathbf{d}_i)_{i=0}^n\) and vice versa. Henceforth, I shall consider the time-series 
\begin{align}
  \label{eq:timeseries}
  (\alpha_i)_{i=0}^n  = \left(\bra{a} X^i \otimes_s X^{i} \ket{\gamma} \right)_{i=0}^n  \qquad &\text{and} \qquad (\beta_i)_{i=0}^n = \left(\bra{b} X^i \ket{d} \right)_{i=0}^n\\
  \intertext{which may be obtained from the actual measurement statistics }
 \left(\bra{a}  \ket{\gamma_i} \right)_{i=0}^n \qquad &\text{and} \qquad \left(\bra{b} \ket{d_i} \right)_{i=0}^n.
\end{align}
Note that this is not simply equivalent to considering the channel \((X,0)\), as the latter is constrained by the complete positivity condition to have symplectic \(X\). \\

The rest of this section captures the idea that the orbit of a state under homogenous time evolution is uniquely determined by the state at a finite number of points in time. This is first demonstrated for discrete evolution with reference to the minimal eliminating polynomial of \(X\). 

\subsection{Time series extensions}
Any endomorphism \(T\) on a finite \(d\)-dimensional vector space has a minimal polynomial of degree \(\delta (T) \leq d\) i.e. a polynomial \(p\) of minimal degree satisfying \(p(T) = 0\). Denoting by \(j_{\lambda}\) the size of the largest Jordan block corresponding to the eigenvalue \(\lambda\) and setting \(j_0 := 0\) if \(T\) has trivial kernel, the minimal polynomial takes the form \(p(x) = \prod_{\lambda} (x - \lambda )^{j_{\lambda}}\). As noted in \cite{rall2025quantum}, this has the consequence that given a finite sequence \((T^i)_{i=t_0}^{t}\) of integer powers of \(T\) with \(t = t_0 + \delta - \text{min}\{t_0, j_0\}\) one can express \(T^{t+1}\) as a linear combination of the terms of this sequence and by induction the sequence may be extended to all future times.
\begin{equation}
  \label{extendforward}
T^{\delta+k} = \sum_{i=0}^{\delta - j_0 - 1}b_i T^{j_0 + i + k}, \text{for all } k \in \mathds{N}_0\\
\end{equation}
Similarly, if \(t_0 > j_0\) the sequence can be extended backward until \(j_0\) according to
\begin{equation}
  \label{extendbackward}
T^{t_0 - k} = \sum_{i=0}^{\delta - j_0 - 1}c_iT^{t_0 + i - k + 1}, \text{for all } k \in \{1, ..., t_0 - j_0\},\\
\end{equation}
where the constants \(b_i, c_i\) are determined by the minimal polynomial. By induction, one can obtain \((T^i)_{i=\text{min}\{t_0, j_0\}}^{\infty}\), where each term in the forward (backward) extension is a linear combination of preceeding (following) terms. This can be applied to the time series in eq. \eqref{eq:timeseries}.
 \begin{theorem}\label{thm:extension}
   For any homodyne measurement \(\mathbf{b}^T\mathbf{R}\) and almost any Gaussian channel\((X,0)\) and state \((\Gamma,d)\) the following hold:
   \begin{enumerate}[(i)]
     \item The time-series 
           \begin{equation}
       (\alpha_i)_{i=t_0}^{t_0 + m(2m+1)} := \left(\bra{a} (X^i \otimes_s X^i) \ket{\gamma} \right)_{i = t_0}^{t_0 + m(2m+1)}\\
     \end{equation}
           fully determines \(\alpha_i\) for all \(i \geq \min\{j_0(X \otimes_s X), t_0\}\).
     \item The time-series
           \begin{equation}
             (\beta_i)_{i=t_0}^{t_0 + 2m} := \left( \bra{b} X^i \ket{d} \right)_{i=t_0}^{t_0 + 2m}.
           \end{equation}
           fully determines \(\beta_i\) for all \(i \geq \min\{j_0(X), t_{0}\}\).
   \end{enumerate}
\end{theorem}
\begin{proof}
  \phantom{.}
  \begin{enumerate}[(i)]
          \item It is straightforward to subsitute \(X \otimes_s X\) for \(T\) above to see that the sequence \(((X \otimes_s X)^i)_{i=t_0}^{t_0 + t}\) with \(t = \delta(X \otimes_s X) - \min\{t_0, j_0(X \otimes_s X)\}\) determines \((X \otimes_s X)^i\) for all \(i \geq \min\{j_0(X \otimes_s X), t_0\}\). The map \((X \otimes_s X)^i \mapsto \alpha_i\) is linear so it follows that each \(\alpha_i\) may be determined from \(((X \otimes_s X)^i)_{i=t_0}^{t_0 + t}\) by a sequence of linear maps. By Cayley-Hamilton, \(\delta(X \otimes_s X) \leq m(2m+1)\).
          \item The argument is identical.
  \end{enumerate}
\end{proof}

\section{One-parameter semigroups}
A family of \(m\)-mode Gaussian channels \((X_t, Y_t)_{t\geq 0}\) forms a quantum dynamical semigroup (QDS)--corresponding to evolution by a time-independent Markovian master equation--if and only if it can be expressed as \cite{heinosaari2009semigroup}
\begin{align}
X_t &= e^{t(A-H)\Omega},\label{eq:gen}\\
      Y_t &= 2 \int_0^t X_s^T \Omega^TB \Omega X_s ds,
\end{align}
with real matrices \(A, B, H \in \mathbb{M}_{2m}(\mathds{R})\) satisfying \(A^T = -A\), \(H^T = H\), \(iA + B \geq 0\). Thus each family of \(X_t\) is generated by a real matrix \(C := (A-H)\Omega \in \mathbb{M}_{2m}(\mathds{R})\) and since every real matrix can be uniquely decomposed in this way any Gaussian QDS can be parameterized by a pair \((C,B) \in \mathbb{M}_{2m}(\mathds{R}) \times \mathbb{S}_{2m}(\mathds{R})\) satisfying
\(B - \frac{i(C-C^{T})\Omega}{2} \geq 0\). In an abuse of nomenclature, the pair \((C,B)\) will be said to generate the corresponding QDS. Denote by \(\mathcal{G}_G\) the set of these. The following proposition defines a measure so that one can refer to null sets within \(\mathcal{G}_G\). 

\begin{proposition}\label{qdsmeas}
  Define isomorphisms \(\nu : \mathbb{M}_{2m}(\mathds{R}) \mapsto \mathds{R}^{4m^2}\) and \(\mu : \mathbb{S}_{2m}(\mathds{R}) \mapsto \mathds{R}^{m(2m + 1)}\). The Lebesgue measure on \(\mathds{R}^{4m^2}\) and \(\mathds{R}^{m(2m+1)}\) induces a measure on \(\mathbb{M}_{2m}(\mathds{R})\) and \(\mathbb{S}_{2m}(\mathds{R})\) respectively. The restriction of the product measure \(\nu \times \mu\) to \(\mathcal{G}_G\) then defines a new measure which is easily seen to be non-trivial since the condition \(B - \frac{i(C-C^{T})\Omega}{2} \geq 0\) defines for every \(C \in \mathbb{M}_{2m}(\mathds{R})\) a convex cone of admissable \(B\)'s which has positive measure in \(\mathbb{S}_{2m}(R)\). \end{proposition}

At time \(t\), an initial state \((\Gamma, \mathbf{d})\) has evolved to \((\Gamma_t , \mathbf{d}_t ) = (X_t \Gamma X_t^T + Y_{t},\; X_t \mathbf{d})\). If \(C\) and \(B\) are known, \(Y_t\) can be calculated for every \(t\) and so as in the case of discrete evolution, it suffices to ignore \(Y\) and consider only \((\Gamma_t'  , \mathbf{d}_t' ) = (X_t \Gamma X_t^T ,\; X_t \mathbf{d})\).  \\

\begin{lemma}\label{clemma}
The set of Gaussian QDS's for which either \(C\) or \(\tilde{C} := 2C \otimes_s I\) is degenerate is a null set. 
\end{lemma}
\begin{proof}
  \(det(\Omega) = 1\) so \(C=(A-H)\Omega\) has degenerate spectrum only if and only if the same is true of \(A-H\). Since \(A-H\) defines a one-to-one correspondence between pairs of \(A=-A^T,H=H^T \in \mathbb{M}_{2m}(\mathds{R})\) and the elements of \(\mathbb{M}_{2m }(\mathds{R})\) itself, and the set of real matrices having degenerate spectrum is a null set, a generic choice of \(A, H\) results in \(C\) having \(2m\) distinct eigenvalues. The additional condition \(iA + B \geq 0\) does not restrict \(A\) to a null set and so in particular it cannot restrict to the null set of \(A,H\) leading to degenerate eigenvalues.\\

  If \(C\) has spectrum \(\{\lambda_1, \hdots , \lambda_{2m}\}\), then the spectrum of \(\tilde{C}\) is \(\{\lambda_i + \lambda_j | \{i,j\} \in \mathcal{I}_{2m}\}\) which is degenerate if \(\lambda_i + \lambda_j - \lambda_k - \lambda_{l} = 0\) for some \(\{i,j\} \neq \{k,l\}\). This is an algebraic variety which must either be zero everywhere or a null set, and the former is obviously false.\\
\end{proof}
The following Lemma is a restatement of a result from \cite{rall2025quantum} for the present context.
\begin{lemma}
  \label{cayleylemma}
  For any \(t \geq 0\) the following holds:
\begin{enumerate}[(i)]
\item The set of \(C \in \mathbb{M}_{2m}(\mathds{R})\) for which \(X_t := e^{tC}\) cannot be expressed as a linear combination of \(\{X_{t_k}\:|\: k=1,\hdots, 2m\}\) is a null set in \(\mathbb{M}_{2m}(\mathds{R})\). 
\item The set of \(C \in \mathbb{M}_{2m}(\mathds{R})\) for which \(\tilde{X}_t := e^{t(2C \otimes_s I)}\) cannot be expressed as a linear combination of \(\{\tilde{X}_{t_k}\:|\: k=1,\hdots, {m(2m + 1)}\}\) is a null set in \(\mathbb{M}_{2m}(\mathds{R})\). 
\end{enumerate}
\end{lemma}
\begin{proof}
  By lemma \ref{clemma}, \(C\) can be assumed to have \(2m\) distinct eigenvalues. Therefore the Cayley-Hamilton theorem can be applied to state that any \(X_t = e^{tC}\) can be expressed as a polynomial in \(C\) of finite degree:
\begin{equation}
  e^{tC} = \sum_{i=0}^{2m-1} y_i(t)C^i \quad \text{with} \quad y_i(t) = \sum_{j=1}^{2m}V_{i,j}e^{\mu_j t}
  \label{eq:ch}
\end{equation}
and \(V\) the inverse of a confluent Vandermonde matrix \cite{hedemann2017explicit} depending on the eigenvalues \(\mu_j\) of \(C\). So \(2m\) integer powers of \(C\) fully determine any \(X_t\).  Supposing there exist functions \(\beta(t)\) through which \(X_t\) can be expressed in terms of \(X_{t_k}\) for a set of fixed times \(t_k\) as \(\sum_{k =1}^{2m} \beta_k(t) X_{t_k} \), then substituting in both sides of \eqref{eq:ch} gives
\begin{equation*}
\begin{split}
\sum_{i=0}^{2m - 1} y_i(t)C^i &= \sum_{i=0}^{2m - 1} \sum_{k=1}^{2m} \beta_k(t)  y_i(t_k) C^i\\
\sum_{i=0}^{2m - 1}  \sum_{j=1}^{2m} V_{ij} e^{\mu_j t}C^i &= \sum_{i=0}^{2m - 1} \sum_{j=1}^{2m} \beta_k(t) \sum_{k=1}^{2m} V_{ij} e^{\mu_j t_k} C^i\\
\end{split}
\end{equation*}
which implies that for \(\forall j, k \in \{1,\hdots , n\}\),
\begin{equation*}
  e^{\mu_j t }  = e^{\mu_j t_{k}} \beta_{k}(t).
\end{equation*}
The desired \(\beta(t)\)'s  can be determined if the matrix \([Z_C]_{j, k} = e^{\mu_j t_k} \in \mathbb{M}_{2m }(\mathds{R})\) is invertible. This is a generalization of a Vandermonde matrix, which is always invertible for real \(e^{\mu_j}\) and distinct \(t\) \cite{yang2001generalization}. This is no longer true for complex eigenvalues, but those \(C\) for which it is not invertible can be shown to be a null set: distinct eigenvalues depend analytically on matrix entries so that the map \(C \mapsto \det(Z_C)\) is analytic. Since it is not constant zero, its zero-set is a null set \cite{mityagin2015zero}. The proof of (ii) follows nearly identically if one notes that the eigenvalues of \(\tilde{C}\) depend analytically on those of \(C\) and that \(\tilde{C}\) can be assumed to have \({m(2m+1)}\) eigenvalues by lemma \ref{clemma}.
\end{proof}

The proof of \cite{rall2025quantum} Thm 3. may be adapted to the present scenario to give the following theorem:
\begin{theorem}\label{thm:semigroup}
Fix a sequence \((t_i)_{i=1}^n\) of arbitrarily spaced times \(t_i \in \mathds{R}^+\). For every pair \((C,B) \in \mathcal{G}_G\) generating a QDS except for a null set in \(\mathcal{G}_G\) the following hold:
\begin{enumerate}[(i)]
\item The sequence
\(        (\alpha_{t_i})_{i=1}^n := \left(\bra{a} (X_{t_i} \otimes_s X_{t_i}) \ket{\gamma} \right)_{i = 1}^{m(2m+1)}
 \)        determines \(\alpha_{t_i}\) for all \(t_i \in \mathds{R}^+\).
        \item The sequence
\(        (\beta_{t_i})_{i=1}^n := \left(\bra{b} X_{t_i}  \ket{d} \right)_{i = 1}^{2m}
 \)        determines \(\beta_{t_i}\) for all \(t_i \in \mathds{R}^+\).
\end{enumerate}
\end{theorem}

\begin{proof}
  In order to show (i), first note that
\begin{equation*}
X_t \otimes X_t = e^{t(C \otimes I + I \otimes C)}
\end{equation*}
therefore 
\begin{equation*}
\begin{split}
2X_t \otimes_s X_t &=  Q e^{t(C \otimes I + I \otimes C) } Q^T \\
            &=  \sum_{k=0}^{\infty} \frac{t^k}{k!} Q(C \otimes I + I \otimes C)^k Q^T\\
            &= \sum_{k=0}^{\infty} \frac{(2t)^k}{k!} (\frac{1}{2}Q(C \otimes I + I \otimes C) Q^T )^k \\
            &=  e^{2t(C \otimes_s I)}.
\end{split}
\end{equation*}
For notional convenience denote \(\tilde{X}_t = 2 X_t \otimes_s X_t \) and \(\tilde{C} = 2 C \otimes_s I\), then
\begin{equation}
\tilde{X}_t = e^{t\tilde{C}}.
\end{equation}
Lemma \ref{cayleylemma} states that for all \(C\) except for a null set, 
\begin{equation}
  \begin{split}
    \tilde{X}_t &= \sum_{k=1}^{m(2m+1)} \left(\sum_{j=1}^{m(2m+1)} (Z_{\tilde{C}}^{-1})_{j,k} e^{\mu_j t}\right) \tilde{X}_{t_k}\\
    \bra{a} X_t \otimes_s X_t \ket{\gamma} &= \sum_{k=1}^{m(2m+1)} \left(\sum_{j=1}^{m(2m+1)} (Z_{\tilde{C}}^{-1})_{j,k} e^{\mu_j t}\right) \bra{a} X_{t_k} \otimes_s X_{t_k} \ket{\gamma}.
  \end{split}
\end{equation}
In order to determine the outcome of the measurement at any time \(t\) it is therefore sufficient to measure \(\mathbf{b}^T\mathbf{R}\) at just \(m(2m+1)\) different arbitrarily spaced times \(t_k\). Following an identical argument, \(X_t = e^{tC}\) can be expressed in terms of \(X_{t_k}\) for the same set of fixed times \(t_k\):
\begin{equation}
\begin{split}
    X_t &= \sum_{k=1}^{2m} \left(\sum_{j=1}^{2m} (Z_C^{-1})_{j,k} e^{\mu_j t}\right) X_{t_k}\\
    \bra{b} X_t \ket{d} &= \sum_{k=1}^{2m} \left(\sum_{j=1}^{2m} (Z_C^{-1})_{j,k} e^{\mu_j t}\right) \bra{b} X_{t_k} \ket{d},
\end{split}
\end{equation}
though only \(2m\) times are needed, to match the dimension of \(C\).\\

Finally, note that any null set of \(C\)'s corresponds to a null set of \(\mathcal{G}_G\) w.r.t. the restriction of the product measure defined in Proposition \ref{qdsmeas}.
\end{proof}

\section{Tomography}

Suppose that one is given independent copies of an unknown \(m\)-mode Gaussian state \((\Gamma,d)\) along with a fixed single-mode homodyne measurement \(\mathbf{b}^T \mathbf{R}\), and a known discrete-time Gaussian evolution \((X,Y)\) which may be applied \(i\) times before performing the measurement.

\begin{theorem}\label{thm:tomography}
  Define the maps \(\alpha: \mathbb{S}_{2m}(\mathds{R}) \rightarrow \mathds{R}^{m(2m+1)}\) such that \\
  \begin{equation}\label{eq:covseries}
  \alpha(\Gamma) :=  \left(\bra{a} (X^i \otimes_s X^i) \ket{\gamma} \right)_{i=t_0}^{t_0 + m(2m+1) - 1}
\end{equation}
and \(\beta: \mathds{R}^{2m} \rightarrow \mathds{R}^{2m}\) such that\\
\begin{equation}
\beta(\mathbf{d}) := \left( \bra{b} X^i \ket{d} \right)_{i = t_0}^{t_0 + 2m - 1}
\end{equation}
Both maps are injective for any \(X \in \mathbb{M}_{2m }(\mathds{R})\) except for a null set depending on \(\mathbf{b}\). 
Hence GDQST with a fixed single-mode homodyne measurement performed at \(m(2m+1)\) consecutive time steps is generically sufficient to uniquely determine a multi-mode Gaussian state.
\end{theorem}
\begin{proof}
  First note that by Theorem \ref{thm:extension}, these maps may be extended backward to \(\min\{j_0(X \otimes_s X), t_0\}\) and \(\min\{j_0(X), t_0\}\) respectively. The set of \(X \in \mathbb{M}_{2m}(\mathds{R})\) having trivial kernel corresponds to \(det(X) = 0\) and is therefore a null set so that one can assume \(j_0(X) = 0\). If \(X\) has non-zero eigenvalues \(\{\lambda_1, \hdots , \lambda_n\}\) then \(X \otimes_s X\) has non-zero eigenvalues \(\{\lambda_i \cdot \lambda_j |\{i,j\} \in \mathcal{I}_{2m}\}\) so it can be assumed that \(j_0(X \otimes_s X) = 0\). Hence, given a sufficiently long initial \enquote*{seed} one can restrict to \(t = 0, ..., m(2m+1) - 1\) for \(\alpha\) and \(t = 0, \hdots, 2m - 1\) for \(\beta\).\\

Measuring \(\var(\mathbf{b}^{T}R)\) after an \(i\)-fold application of the channel \((X,Y)\) corresponds to a dual vector 
%\(\bra{11}(P \otimes P)(\Lambda^i \otimes \Lambda^i)(P^{-1} \otimes P^{-1})\) acting on \(\ket{\gamma}\)
\(\bra{a} (X^i \otimes_s X^i)\) acting on the symmetrically vectorized covariance matrix \(\ket{\gamma}\).
The set of real matrices which are not diagonalizable form a null set, so write \(X^i = P \Lambda^i P^{-1} \), with \(\Lambda = \sum_{\mu=1}^{2m}\lambda_\mu \ketbra{\mu} \in \mathds{C}^{2m\times 2m}\) diagonal and \(P \in \mathbb{M}_{2m }(\mathds{R})\). Then
\begin{equation*}
    (X \otimes_s X)^i = P \otimes_s P (\Lambda \otimes_s \Lambda)^i P^{-1} \otimes_s P^{-1}.
\end{equation*}
Absorbing the diagonalizing matrices as \(\bra{a' } := \bra{a} P \otimes_s P \) and \(\ket{\gamma' } := P^{-1} \otimes_s P^{-1} \ket{\gamma}\), one is left with
\begin{equation}
  (\Lambda \otimes_s \Lambda)^i = \sum_{\{\mu,\nu\} \in \mathcal{I}_{2m}} (\lambda_{\mu} \cdot \lambda_{\nu})^i \ketbra{\psi_{\mu,\nu}}.\\
\end{equation}
The map \(\alpha(\Gamma)\) can then be written as
\begin{equation}
  \begin{split}
          \alpha: \Gamma \mapsto & \sum_{i=0}^{m(2m+1) - 1} \bra{a' } (\Lambda^i \otimes_s \Lambda^i) \ket{\gamma '} \ket{i}\\
          &= \sum_{i=0}^{m(2m+1) - 1}\sum_{\{\mu,\nu\} \in \mathcal{I}_{2m}} (\lambda_{\mu} \cdot \lambda_{\nu})^i \braket{a' }{\psi_{\mu, \nu}} \braket{\psi_{\mu,\nu}}{\gamma '} \ket{i}\\
    &= MN \ket{\gamma' }
  \end{split}
\end{equation}
where \([N]_{ii} := \braket{a' }{\psi_{i}} \in \lor^2 \mathds{R}^{2m}\) is diagonal and
\([M]_{i,\{\mu,\nu\}} := (\lambda_{\mu} \cdot \lambda_{\nu})^i \in \lor^2 \mathds{R}^{2m}\) is a Vandermonde matrix. \\

To see that the map is invertible almost everywhere, first note that
\begin{equation}\label{eq:detN}
\text{det}(N) = \prod_{\{\mu, \nu\}} \bra{a} P \otimes_s P \ket{\psi_{\mu,\nu}}.
\end{equation}
Since \(P \otimes_s P\) is invertible, the vectors \( P \otimes_s P \ket{\psi_{\mu,\nu}}\) form a basis for \(\lor^{2} \mathds{R}^{2m}\) and hence the \(\ket{a}\) orthogonal to one or more of these belong to a null set. Next, note that\(M\) has determinant
\begin{equation}\label{det}
\text{det}(M) = \prod_{\{\mu,\nu\} < \{\mu' , \nu' \}}(\lambda_{\mu} \cdot \lambda_{\nu} - \lambda_{\mu' } \cdot \lambda_{\nu' })
\end{equation}
where \(<\) is with respect to lexicographic ordering. The zero set hereof is an algebraic vareity which either fills the entire space or is a manifold of one fewer dimension and hence of Lebesgue measure zero. To exclude the possibility of it being zero everywhere, suppose that the \(i\)-th eigenvalue is the \(i\)-th prime, then \(\lambda_{\mu} \cdot \lambda_{\nu} - \lambda_{\mu' } \cdot \lambda_{\nu' } = 0 \iff \lambda_{\mu} = \lambda_{\mu' } \text{  and  } \lambda_{\nu} = \lambda_{\nu' }\) which cannot be since \(\{\mu,\nu\} < \{\mu' , \nu' \}\).\\

Similarly, the displacement vector is acted upon by a dual vectors \(\bra{1}X^i\). Denoting \(\bra{\varphi} := \bra{1}P \) and \(\ket{d' } := P^{-1} \ket{d}\),
\begin{equation}
  \begin{split}
    \ket{d} \mapsto &\sum_{i=0}^{2m - 1} \bra{1} X^i \ket{d}\ket{i}\\
    &= \sum_{i=0}^{2m - 1} \bra{\varphi} \Lambda^i \ket{d' }\ket{i}\\
    &= \sum_{i=0}^{2m - 1} \sum_{j=1}^{2m} (\lambda_j)^i \braket{\varphi}{j}\braket{j}{d' }\ket{i}.
  \end{split}
\end{equation}
Again, this yields \(\ket{d} \mapsto AB \ket{d' }\) with a Vandermonde matrix \([A]_{ji} = (\lambda_j)^i\) having \(\det(A) = \prod_{j < j' }(\lambda_j - \lambda_{j' })\) and a diagonal matrix \([B]_{jj} = \braket{\varphi}{j}\). A nearly identical argument shows that these are invertible for all \(X\) and \(\mathbf{b}\) excluding a null set.
\end{proof}

\begin{corollary}
  The maps
  \begin{equation}\label{eq:contcovseries}
\mathbb{S}_{2m}(\mathds{R}) \ni \Gamma \mapsto (\bra{a} X_{t_k} \otimes_s X_{t_k} \ket{\gamma})_{k=1}^{m(2m+1)}
  \end{equation}
  and
  \begin{equation}
\mathds{R}^{2m} \ni \mathbf{d} \mapsto (\bra{b} X_{t_k} \ket{d})_{k=1}^{2m}
   \end{equation}
   are invertible for almost any Gaussian channel \((X_t, Y_t)_{t \geq 0}\) forming a quantum dynamical semigroup.
\end{corollary}
\begin{proof}
  This follows directly from applying Theorem \ref{thm:extension} to see that a finite time series \(t_1, \hdots, t_{m(2m+1)}\) (or \(t_1, \hdots, t_{2m}\), respectively) can be extended to all positive times and in particular to the times \(0, \hdots, m(2m+1) -1\) (or \(0, \hdots 2m -1\), respectively) for which Theorem \ref{thm:tomography} applies.
\end{proof}

Theorem \ref{thm:tomography} states that the set of evolutions for which this tomography scheme fails is a null set. It turns out however that some of the evolutions in this null set are physically relevant.
\begin{corollary}
  If the Gaussian channel is unitary, then neither the map \(\alpha\) defined in \eqref{eq:covseries} nor it's continuous-time version defined in \eqref{eq:contcovseries} is invertible and GDQST fails to indentify arbitrary initial states.
\end{corollary}
\begin{proof}
  Consider the case of a Gaussian unitary, which has \(X\) simplectic and \(Y=0\). The symplectic condition \(X^T \Omega X = \Omega\) implies \(\Omega X^{-1}= X^T \Omega\). From this it follows that if \(\lambda\) is an eigenvalue of \(X\), then \(\lambda^{-1}\) is also an eigenvalue \footnote{Since if \(X\mathbf{v} = \lambda \mathbf{v}\) then \(\mathbf{v}^TX^T \Omega = \lambda \mathbf{v}^T \Omega \) so that \(\mathbf{v}^T \Omega X^{-1} = \lambda \mathbf{v}^T \Omega\) and hence \(\lambda^{-1} (\mathbf{v}^T \Omega) = (\mathbf{v}^T \Omega) X\)}.  Recalling the expression for \(\det (M)\) in eq. \eqref{det} one finds factors \((\lambda_{\mu} \cdot \lambda_{\mu}^{-1} - \lambda_{\mu' } \cdot \lambda_{\mu' }^{-1}) = 0\) and hence \(\alpha\) is not invertible. \\
  For Gaussian unitary evolution according to a semigroup \((X_t, Y_t)_{t\geq 0}\),  Theorem \ref{thm:semigroup} states that a vector constructed from any number of \(\bra{a} \tilde{X}_t \otimes_s \tilde{X}_t \ket{\gamma}\) can always be written as a linear transformation of \(\alpha(\Gamma)\) with \(X = X_1\) unitary. Since \(\alpha\) is a non-invertible, neither is the composition. 
\end{proof}

Obviously, a single-mode measurement is not sufficient to recover a multi-mode state if there is insufficient mixing (or swapping) of modes:
\begin{corollary}
  If \(X\) is block-diagonal (i.e. there exists at least a bipartition of the set of modes such that no mixing occurs between the partitions), then a single-mode measurement is insufficient for GDQST of the entire multi-mode state.
\end{corollary}
\begin{proof}
  This follows by observing that the diagonalizing matrix \(O\) inherits the block structure so a measurement of, say, the 1st mode, leads to \(\det(B) = \prod_{j=1}^{2m} \bra{1}O\ket{j} = 0\) because \(O^T\ket{1}\) must be in the subspace defined by the image of the first block, which is necessarily orthogonal to some \(\ket{j}\) by the assumption that there are at least two blocks. Hence the displacement vector cannot be identified and GDQST is not possible.
\end{proof}

\subsection{Pure states}
\begin{proposition}
  The covariance matrix of an \(m\)-mode pure Gaussian state takes the form
  \begin{equation}\label{eq:purecov}
    \Gamma = \begin{pmatrix}
      A & AB \\
      BA & BAB + A^{-1}
    \end{pmatrix}
  \end{equation}
  with real symmetric \(m \times m\) matrices \(A\) and \(B\) and \(A > 0\).
\end{proposition}
\begin{proof}
  It is well known that if the covariance matrix is written in block form as \( \Gamma = \begin{pmatrix}
  A & C\\
  C^T & D,
\end{pmatrix}\)
with \(A, D \in \mathbb{S}_m\) and \(C \in \mathbb{M}_{m}(\mathds{R})\) then the symplectic condition \((\Gamma \Omega)^2 = -\mathds{1}\) yields \(AD = \mathds{1} + C^2,
  DC = (DC)^T,
  CA = (CA)^T\). The result follows from relabelling \(B = A^{-1}C\).
\end{proof}

\begin{theorem}\label{thm:puretomo}
  For a pure Gaussian state \((\Gamma, d), (\Gamma\Omega)^2 = -\mathds{1}\), the map 
  \begin{equation}\label{eq:alphapure}
    \Gamma \mapsto (\bra{a} X^i \otimes_s X^i \ket{\gamma})_{i = 0}^{m(m+1) - 1}\\
  \end{equation}
  is injective for any \(X\) other than a null set depending on the measurement \(\mathbf{b}^T\mathbf{R}\). Hence GDQST with a fixed single-mode homodyne measurement performed at \(m(m+1)\) consecutive time steps is generically sufficient to identify any pure multi-mode Gaussian state.
\end{theorem}
\begin{proof}
  Expressing \(\Gamma\) as in eq. \eqref{eq:purecov} and counting parameters, gives \(\frac{m(m+1)}{2}\) for each of \(A\) and \(B\) for a total dimension of \(m(m+1)\) for pure Gaussian states (compared to \(m(2m+1)\) for unconstrained Gaussian states).
\(\ket{\gamma}\) can be expressed in the same basis of the symmetric subspace introduced earlier as
\begin{equation}\label{eq:gammaexpanded}
  \begin{split}
  \ket{\gamma} = &\sum_{\substack{\mu ,\nu =1\\\mu  \leq \nu }}^{m}(\sqrt{2}(1 - \delta_{\mu\nu}) + \delta_{\mu \nu })A_{\mu \nu }\ket{\psi_{\mu,\nu}}\\
            &+ \sum_{\mu ,\nu =1}^m \sqrt{2} (AB)_{\mu \nu }\ket{\psi_{\mu ,\nu +m}}\\
  &+ \sum_{\substack{\mu ,\nu =1\\ \mu \leq \nu }}(\sqrt{2}(1 - \delta_{\mu\nu}) + \delta_{\mu \nu })(BAB + A^{-1})_{\mu \nu }\ket{\psi_{\mu +m,\nu +m}}.
  \end{split}
\end{equation}
  Noting that \(B = A^{-1}(AB)\) and that \(AB\) is symmetric by the symplectic condition, the parameterization \((A,B) \mapsto \begin{pmatrix} A & AB\\ BA & BAB + A^{-1} \end{pmatrix}\) yields an injective map \(f:  \mathbb{S}_{m(m+1)/2}\times  \mathbb{S}_{m(m+1)/2} \rightarrow \mathds{R}^{m(2m+1)} \) such that \((A, AB) \mapsto \ket{\gamma}\). Hence the coefficients of \(\mathcal{I}_{\text{pure}} := \{\ket{\psi_{\mu,\nu}} |  1 \leq \mu \leq \nu \leq m\} \cup \{\ket{\psi_{\mu,\nu +m}} | 1 \leq \mu \leq \nu \leq m\}\) in \eqref{eq:gammaexpanded}
are sufficient to reconstruct \(\ket{\gamma} = f(A, AB)\). The rest of the proof follows that of Theorem \ref{thm:tomography}, with the caveat that the time series is not of sufficient length to apply Theorem \ref{thm:extension} so \(t_0 = 0\) by necessity. Define a map \begin{equation}
  \Gamma \mapsto \sum_{i=0}^{m(m+1)-1} \sum_{\{\mu,\nu\} \in \mathcal{I}_{\text{pure}}} (\lambda_{\mu} \cdot \lambda_{\nu})^i \braket{a' }{\psi_{\mu,\nu}} \braket{\psi_{\mu,\nu}}{\gamma' }\ket{i}
\end{equation}
which corresponds to the product of a vandermonde matrix and a diagonal matrix as before, but now acting on \(\text{span}(\mathcal{I}_{\text{pure}})\). An identical argument convinces of their invertibility so \(\bra{\psi_{\mu,\nu}}\ket{\gamma}\) can be determined for \(\forall \{\mu, \nu\} \in \mathcal{I}_{\text{pure}}\) and the map \eqref{eq:alphapure} is injective. Since \(m(m+1) \geq 2m, \:\forall m \geq 1\), the case of the displacement vector is already covered by Theorem \ref{thm:tomography}.

\end{proof}

\noindent \emph{Acknowledgements}\\
The author acknowledges funding from the MCQST and the Munich Quantum Valley, which is supported by the Bavarian state government through the Hightech Agenda Bayern Plus.

\bibliographystyle{halpha}
\bibliography{refs}
\end{document}